\documentclass[10pt]{book}
\usepackage{array}
\usepackage[latin1]{inputenc}
\usepackage[dvips]{graphicx}
\usepackage{dsfont,amssymb,amsmath,a4wide,fancyhdr,tabularx}
 \renewcommand{\baselinestretch}{1.1}
\newcommand{\abs}[1]{\left\vert\:  #1 \:\right\vert}
\newcommand{\set}[1]{\left\lbrace\:  #1 \:\right\rbrace}

\newtheorem{Theorem}{\bf Theorem}[chapter]
\newtheorem{Remark}{\bf Remark}[chapter]

\newcommand{\R}{\mathbb{R}}
\newcommand{\E}{\mathbb{E}}

\newcommand{\U}{\mathcal{U}}

\newcommand{\eff}{\mathrm{eff}}
\DeclareMathOperator{\conv}{conv}

 \usepackage{titlesec}
 \newcommand{\qed}{\nobreak \ifvmode \relax \else
      \ifdim\lastskip<1.5em \hskip-\lastskip
      \hskip1.5em plus0em minus0.5em \fi \nobreak
      \vrule height0.75em width0.5em depth0.25em\fi}
\newenvironment{proof}[1][\sc Proof:]{\begin{trivlist}\item[\hskip \labelsep {\bfseries #1}]}{\end{trivlist}}

\begin{document}


\renewcommand{\baselinestretch}{1.2}

\markright{
\hbox{\footnotesize\rm Statistica Sinica (2007): Preprint}\hfill
}

\markboth{\hfill{\footnotesize\rm HOLGER DETTE AND CHRISTINE KISS} \hfill}
{\hfill {\footnotesize\rm OPTIMAL DESIGNS FOR INVERSE REGRESSION MODELS} \hfill}

\renewcommand{\thefootnote}{}
$\ $\par


\fontsize{10.95}{14pt plus.8pt minus .6pt}\selectfont
\vspace{0.8pc}
\centerline{\large\bf Optimal experimental designs for inverse}
\vspace{2pt}
\centerline{\large\bf quadratic regression models}
\vspace{.4cm}
\centerline{Holger Dette, Christine Kiss}
\vspace{.4cm}
\centerline{\it Ruhr-Universität Bochum}
\vspace{.55cm}
\fontsize{9}{11.5pt plus.8pt minus .6pt}\selectfont


\begin{quotation}
\noindent {\it Abstract:}
In this paper optimal experimental designs for inverse quadratic
regression models
 are determined. We consider
 two  different parameterizations of the model and
 investigate local optimal designs with respect to the  $c$-, $D$- and $E$-criteria,
 which reflect various  aspects  of the precision of
  the maximum likelihood estimator for the parameters in inverse quadratic regression models.
In particular it is demonstrated that for a sufficiently large
 design space  geometric allocation rules are optimal with respect to many optimality criteria.
 Moreover, in numerous cases the designs with respect to the different criteria
 are supported at the same points. Finally,
  the efficiencies of different optimal designs with respect to various
  optimality criteria are studied, and the efficiency of
  some commonly used designs are investigated.
\par

\vspace{9pt}
\par
\end{quotation}\par


\fontsize{10.95}{14pt plus.8pt minus .6pt}\selectfont

\setcounter{chapter}{1}
\setcounter{equation}{0} 
\noindent {\bf 1. Introduction}

Inverse polynomials  define a flexible family of nonlinear regression models which are used to describe the relationship between a response,
say $Y$, and a univariate predictor, say  $u$ [see eg. Nelder (1966)]. The model is defined by the  expected response
\begin{equation}
\label{mod1}
\E(Y|u)=\frac{u}{P_n(u,\theta)},\;u\geq 0,
\end{equation}
where $P_n(u,\theta)$ is a polynomial of degree  $n$ with coefficients $\theta_0,\dots,\theta_n$
 defining the shape of the curve.
Nelder (1966) compared the properties of inverse and  ordinary polynomial models
for analyzing data.
In contrast to ordinary polynomials inverse polynomial regression
models are  bounded and can be used to describe  a saturation effect,
in which case the response does not exceed a finite amount. Similarly,
a toxic effect can be produced,
in which case the response eventually falls to zero.

An important class of inverse polynomial models are defined by
 inverse quadratic regression models, which correspond to the case $n=2$ in (\ref{mod1}).
These models have numerous applications, in particular in chemistry and agriculture [see Ratkowski (1990), Sparrow (1979a, 1979b), Nelder (1960), Serchand, McNew,
Kellogg and Johnson (1995) and Landete-Castillejos and Gallego (2000) among others]. For example,
 Sparrow (1979a, 1979b) analyzed data from
 several series of experiments designed to study the relationship between crop yield
 and fertilizer input. He concluded that among several competing models the inverse quadratic
 model  produced the  best fit to data obtained from
  yields of  barley and grass crops. Similarly,
  Serchand et al. (1995) argued that  inverse polynomials can produce
a dramatically steep rise and might realistically describe lactation curves.

 While much attention has been paid to the construction of various optimal
 designs for the inverse linear or Michaelis Menten  model [see Song and Wong (1998), Lopez-Fidalgo and Wong (2002),
 Dette, Melas and Pepelyshev (2003), Dette and Biedermann (2003), among many others], optimal
 designs for the inverse quadratic regression model have not been studied in so much detail.
 Cobby, Chapman and Pike  (1986) determined local $D$-optimal designs numerically
 and Haines (1992) provided some analytical results for $D$-optimal designs in the inverse
 quadratic regression model. In particular, in these references it is demonstrated that geometric
 allocation rules are $D$-optimal.  The present paper is devoted to a more systematic study
 of local optimal  designs for inverse  quadratic models. We consider the $c$-, $D$-, $D_1$-
 and $E$-optimality criterion  and determine local optimal designs for two different
 parameterizations of the inverse quadratic regression model. In Section 2 we introduce two
 parameterizations of the inverse quadratic regression model and describe some basic facts
 of approximate design theory. In Section 3 we  discuss several $c$-optimal designs.
In particular  $D_1$-optimal designs are determined, which are of particular
 importance if discrimination between an inverse linear  and inverse quadratic model
 is one of the interests of the experiment. As a further special case of the $c$-optimality
 criterion we determine optimal extrapolation designs.
 Section 4 deals with the local $D$-optimality and  $E$-optimality criterion.
  It is shown
 that for all  criteria  under consideration geometric designs are local optimal, whenever the design
 space is sufficiently large.  We also determine the structure of the local optimal designs
 in the case of a bounded design space. These findings extend the observations made by
 Cobby, Chapman and Pike  (1986) and Haines (1992) for the $D$-optimality criterion to other optimality
 criteria, different design spaces and a slightly different
 inverse quadratic regression model.
\par

\bigskip
\setcounter{chapter}{2}
\setcounter{equation}{0} 
\noindent {\bf 2. Preliminaries}

We consider two parameterizations of the  inverse quadratic regression model
\begin{equation}\label{mo1a}
\E(Y|u)= \eta (u, \theta)~,
\end{equation}
where $ \theta = (\theta_0,\theta_1,\theta_2)^T$ denotes the vector
of unknown parameters and   the expected response  is given by
\begin{subequations}\label{parametrisierungen}
    \begin{align}
        \eta_1(u,\theta) =&\frac{u}{\theta_0+\theta_1u+\theta_2u^2}&\;
       \label{para3}\\
\mbox{or } ~~~ ~~~ ~~~ ~~~ ~~~~~~~~~~\nonumber ~~~ & ~~~ ~~~ &\\
        \eta_2(u,\theta)=&\frac{\theta_0 u}{\theta_1+u+\theta_2 u^2}&
         \label{para1}
    \end{align}
\end{subequations}
The explanatory variable varies in the interval  $\U=[s,t]$, where $s \ge 0$ and $0<s<t < \infty $,
or in the unbounded set $\U=[s,\infty) $ with $s \ge 0$.
The assumptions regarding  the parameters vary with the different parameterizations and
 should assure, that the numerator in (\ref{para3}) and (\ref{para1}) is positive on ${\cal U}$.
 Under such assumptions the  regression functions have  no points of discontinuity. Moreover, both functions are strictly  increasing to a maximum of size $(\theta_1+2\sqrt{\theta_0\theta_2})^{-1}$ at the point
 $u_{\max_1}=\sqrt{\theta_0/\theta_2}$ for parameterization (\ref{para3})
  and to a maximum of size $\theta_0(1+\sqrt{\theta_1\theta_2})^{-1}$ at the point
  $u_{\max_2}=\sqrt{\theta_1/\theta_2}$ for parameterization (\ref{para1})  and then the functions are strictly decreasing to a zero asymptote.
 A sufficient condition for the positivity of the numerator is
   $\theta_0,\theta_2>0$, $\abs{\theta_1}\leq 2\sqrt{\theta_0\theta_2}$ for model (\ref{para3})
   and $\theta_0,\theta_1,\theta_2>0,$ $2 \sqrt{\theta_1\theta_2}>1$ for model (\ref{para1}), respectively.
  We assume that at each $u \in \U$ a  normally distributed observation is  available with mean
  $\eta (u,\theta)$ and variance $ \sigma^2 >0$, where the function $\eta $ is either $\eta_1$ or $\eta_2$,
  and different observations are assumed to be independent.
An experimental design $\xi$ is a probability measure with finite support defined on the  set
  $\U$ [see Kiefer (1974)]. The information matrix of an experimental design $\xi$ is defined by
\begin{equation}
\label{infmat}
M(\xi,\theta)=\int_{\U}f(u,\theta)f^T(u,\theta)d\xi(u),
\end{equation}
where
\begin{equation}
\label{grad}
f(u,\theta) = \frac{\partial}{\partial \theta} \eta (u,\theta)
\end{equation}
denotes the gradient of the expected response with respect to the parameter $\theta$. For the two parameterizations (\ref{para3}) and
(\ref{para1}) the vectors of the partial derivatives are given by
\begin{equation}\label{grad3}
    f_{1}(u,\theta)=\frac{-u}{(\theta_0+\theta_1 u+\theta_2 u^2)^2}\left(1,u,u^2\right)^T
 \end{equation}
 and
\begin{equation}\label{grad1}
    f_{2}(u,\theta)=\frac{u}{\theta_1+u+\theta_2 u^2}\left(1,-\frac{\theta_0}{\theta_1+u+\theta_2 u^2},-\frac{\theta_0
  u^2}{\theta_1+u+\theta_2 u^2}\right)^T,
\end{equation}
respectively. \\
If $N$ observations can be made and the design $\xi$ concentrates mass $w_i$ at the points $u_i$, $i=1,\dots,r$,
 the quantities $w_iN$ are rounded to integers such that  $\sum_{j=1}^rn_i=N$ [see Pukelsheim and Rieder (1992)],
and the experimenter takes $n_i$ observations at each point $u_i$, $i=1,\dots,r$. If the sample size $N$ converges to infinity, then (under
appropriate assumptions of regularity)
 the covariance matrix of the maximum likelihood estimator for the parameter  $\theta$ is approximately
 proportional to the matrix $(\sigma^2/N)M^{-1}(\xi,\theta)$, provided that the inverse of the information matrix exists [see Jennrich (1969)].
 An optimal experimental design maximizes or minimizes an appropriate functional of the information matrix or its inverse,
and there are numerous optimality criteria which can be used to discriminate between
 competing designs [see Silvey (1980) or  Pukelsheim (1993)].
 In this paper we will investigate the $D$-optimality criterion, which maximizes the determinant of the inverse of the information matrix
 with respect to the design $\xi$, the $c$-optimality criterion, which minimizes the
 variance of the maximum likelihood estimate for the linear combination $c^T\theta$
and the $E$-optimality criterion, which maximizes the minimum eigenvalue of
 the information matrix $M(\xi,\theta)$.
\par

\bigskip
\setcounter{chapter}{3}
\setcounter{equation}{0} 
\noindent {\bf 3. Local $c$-optimal designs}

Recall that for a given vector $c \in \R^{3}$ a design $\xi_c$ is called $c$-optimal if the linear combination $c^T\theta$ is estimable by the
design $\xi_c$, that is $\mbox{Range}(c) \subset \mbox{Range}(M(\xi_c , \theta))$, and the design $\xi_c$ minimizes
\begin{equation}
\label{ccrit}
c^TM^-(\xi,\theta)c
\end{equation}
among all designs for which $c^T\theta$ is estimable,
where $M^-(\xi,\theta)$ denotes a generalized inverse of the matrix $M(\xi,\theta)$.
It is shown in Pukelsheim (1993) that the expression (\ref{ccrit}) does not depend on the
specific choice of the generalized inverse. Moreover,
a design $\xi_c$ is $c$-optimal if and
only if there exists a generalized inverse $G$ of $M(\xi_c,\theta)$ such that
the inequality
\begin{equation}
(f'(u,\theta)Gc)^2\leq c'M^-(\xi_c,\theta)c
\end{equation}
holds for all $ u\in\U$ [see Pukelsheim (1993)].
 A further important tool to determine $c$-optimal
designs is the theory of Chebyshev systems, which will be briefly described
here for the sake of
completeness.

Following Karlin and Studden (1966) a set of functions $\{g_0,\dots,g_n\}$ defined on the set $\U$ is  called Chebychev-system, if every
linear combination $\sum_{i=0}^na_ig_i(x)$ with $\sum_{i=0}^na_i^2>0$ has at most $n$ distinct roots on $\U$. This property is equivalent to
the fact that
    \begin{equation}\label{defkiefwolf}
        \det(g(u_0),\dots,g(u_n))\neq 0
    \end{equation}
holds for all $u_0,\dots,u_n \in \U$ with $u_i \neq u_j \ (i \neq j)$, where $g(u)=(g_0(u),\dots,g_n(u))^T$ denotes the vector of all functions
[see Karlin and Studden (1966)]. If the functions $g_0,\dots,g_n$ constitute a Chebyshev-system on  the set $\U$, then there exists a unique
``polynomial''
    \begin{equation}\label{eindpolynom}
        \phi(u):=\sum_{i=0}^n \alpha_i^* g_i(u) \ \ \ (\alpha^*_0, \dots, \alpha^*_n \in \mathbb{R})
    \end{equation}
with the following properties
    \begin{equation*}
        \begin{aligned}
            (i)&\qquad  \abs{\phi(u)}\leq 1\quad\forall u\in \cal{U}\\
            (ii)& \qquad \text{There exist } n+1 \text{ points } s_0<\dots<s_n \text{ such that } \phi(s_i)=(-1)^{n-i}\\
            &\qquad\text{ for }
                        i=0,\dots,n.
        \end{aligned}
    \end{equation*}
The function $\phi(u)$ is called the Chebychev-polynomial, and the points $s_0,\dots, s_n$ are called Chebychev-points, which are not
necessarily unique. Kiefer and Wolfowitz (1965) defined the set  $A^* \subset \R^{n+1}$ as the set
 of all vectors $c\in\R^{n+1}$  satisfying
\begin{equation} \label{astern}
\begin{vmatrix}
    g_0(x_1)&\dots &g_0(x_n)&c_0\\
    g_1(x_1)&\dots &g_1(x_n)&c_1\\
    \vdots& &\vdots&\vdots\\
    g_n(x_1)&\dots &g_n(x_n)&c_n\\
\end{vmatrix}\neq 0,
\end{equation}
whenever the points  $x_1,\ldots ,x_n  \in \U$ are distinct. They showed that for each $c\in A^*$
 the $c$-optimal design, which minimizes
 $$
 c^T \left(\int_{\cal{U}} g (u) g^T (u) d \xi (u)\right)^{-1} c
 $$
 among all designs on $\cal{U}$,
  is  supported by the entire set of the Chebychev-points $s_0,\dots,s_n$.
 The corresponding optimal weights  $w_0^*,\ldots ,w_n^*$
 can then easily be found using Lagrange multipliers
 and are given by
\begin{equation}\label{gewichtsvektor}
w_i^*=\frac{|v_i|}{\sum_{j=0}^n|v_j|}\qquad i=0,\dots,n,
\end{equation}
where the vector $v$ is defined by
\begin{equation*}
v=(XX^T)^{-1}Xc,
\end{equation*}
and the $(n+1)\times (n+1)$-matrix $X$ is given by $X=\left(g_{j}(s_{i})\right)_{i,j=0}^n$ [see also Pukelsheim and Torsney (1991)].

In the following discussion we will use these results to determine local optimal design for two specific goals in the data analysis with
inverse quadratic regression models: discrimination between inverse linear and quadratic models and extrapolation  or prediction at a specific
point $x_e$. We will begin with the discrimination problem, which has been extensively studied for ordinary polynomial regression models [see
Stigler (1971), Studden (1982) or Dette (1995), among many others]. To our knowledge the problem of constructing designs for the discrimination
between inverse rational models has not been studied in the literature. We consider the inverse quadratic regression model (\ref{para3}) and
are interested in determining a design, which can be used to discriminate between this and the inverse linear regression model
$$
\eta (u, \theta) = \frac{u}{\theta_0 + \theta_1u}.
$$
The decision, which model should be used could be based on the likelihood ratio test for the hypothesis $H_0: \theta_2 =0$ in the model
(\ref{para3}), and a standard calculation shows that the (asymptotic) power of this test is a decreasing function  of the quantity
(\ref{ccrit}), where the vector $c$ is given by $c=(0,0,1)^T$. Thus a design maximizing the power of the likelihood ratio test for
discriminating between the inverse linear and quadratic model is a local $c$-optimal design for the vector $c=(0,0,1)^T$. Following Stigler (1971) we
call this design local $D_1$-optimal. Our first results determine the local $D_1$-optimal design for the two parameterizations of the inverse
quadratic regression model explicitly.

\medskip
\begin{Theorem}\label{theoD13}
    The local $D_1$-optimal  design $\xi^*_{D_1}$ for the inverse quadratic regression
model (\ref{para3}) on the design space $\U=[0,\infty)$ is given by
        \begin{equation}\label{D1optpar3}
            \xi^*_{D_1}=
                \begin{pmatrix}
                    \frac{1}{\rho}\sqrt{\frac{\theta_0}{\theta_2}}&\sqrt{\frac{\theta_0}{\theta_2}}& \rho \sqrt{\frac{\theta_0}{\theta_2}} \\
                    w_0& w_1& 1-w_0-w_1
                \end{pmatrix}
        \end{equation}
    with weights
        \begin{align*}   w_0&=\Big(\sqrt{\theta_2}\theta_0+\theta_1\sqrt{\theta_0}\rho+\sqrt {\theta_2}\theta_0\rho^2\Big)^2/[(1+\rho)\lambda]\\    w_1&=\Big(2\sqrt{\theta_2}\theta_0+\theta_1\sqrt{\theta_0}\Big)^2\rho^2/\lambda\\
        \end{align*}and $\lambda=\theta_0\Big(\theta_0\theta_2\big(1+6\rho^2+\rho^4\big)+2\theta_1\rho\big(\theta_1\rho+\sqrt{\theta_0\theta_2}\big(1+\rho\big)^2\big)\Big)$.\\
    The geometric scaling factor $\rho$ is defined by
        \begin{equation}\label{scalefac}
            \rho= \rho(\gamma) = 1 + \frac{2 + \gamma}{\sqrt{2}} + \sqrt{2(1+\sqrt{2}) + (2+\sqrt{2})\gamma + \frac{\gamma^2}{2}}
        \end{equation}
    with $\gamma=\theta_1/\sqrt{\theta_0\theta_2}$.
    This design is also local $D_1$-optimal on the design space ${\cal U}=[s,t]$ ($0<s<t$), if the inequalities $0\leq s \leq
    \rho^{-1}\sqrt{\theta_0/\theta_2}$
    and $t \geq \rho \sqrt{\theta_0/\theta_2}$ are satisfied. \\
 The local $D_1$-optimal  design on the design space ${\cal U}=[s,t]$ for
    model (\ref{para3}) is of the form
        \begin{equation}\label{D1fall31}
                      \xi^*_{D_1}=
                  \begin{pmatrix}
                        s&u'_1&u'_2\\
                        w'_0&w'_1&1-w'_0-w'_1
                  \end{pmatrix},
        \end{equation}
if the inequalities $s \geq \rho^{-1}\sqrt{\theta_0/\theta_2}$ and $t > \rho
              \sqrt{\theta_0/\theta_2}$ hold, of the form
              \begin{equation}\label{D1fall32}
              \xi^*_{D_1}=
                    \begin{pmatrix}
                        u''_0&u''_1&t\\
                        w''_0&w''_1&1-w''_0-w''_1
                    \end{pmatrix},
                \end{equation}
if the inequalities $s < \rho^{-1}\sqrt{\theta_0/\theta_2}$ and $t \leq \rho
                \sqrt{\theta_0/\theta_2}$
are satisfied,  and is of the form
                \begin{equation}\label{D1fall33}
                \xi^*_{D_1}=
                    \begin{pmatrix}
                        s&u'''_1&t\\
                        w'''_0&w'''_1&1-w'''_0-w'''_1
                    \end{pmatrix},
                \end{equation}
if the inequalities $s \geq \rho^{-1}\sqrt{\theta_0/\theta_2}$ and $t \leq \rho \sqrt{\theta_0/\theta_2}$ hold.

\end{Theorem}

\begin{proof}
The proof is performed in three steps:
\begin{enumerate}
\item[(A)] At first we identify a candidate for the local $D_1$-optimal  design on the interval $[0,\infty)$ using the theory of Chebyshev polynomials.
\item[(B)] We use the properties of the Chebyshev polynomial (\ref{eindpolynom}) to prove
the local $D_1$-optimality of this candidate.
\item[(C)] We consider the case of a bounded design space and determine how the constraints interfere with the support points of the local optimal design on the unbounded design space.
\end{enumerate}\vspace{2ex}
{\bf (A):} Let $f(u,\theta)$ be the vector of the partial derivatives in parameterization (\ref{para3}) defined in \eqref{grad3}. It is easy to
see, that the components of the vector $f_{1}(u,\theta)$, say
  $\set{f_{10}(u,\theta),f_{11}(u,\theta),f_{12}(u,\theta)}$,
constitute a Chebyshev-system   on any  bounded interval $[s,t] \subset (0, \infty) $. 
Furthermore for $y_0,y_1>0$ with $y_0\neq y_1$ we get
                $$\begin{vmatrix}
                    f_{10}(y_0,\theta)&f_{10}(y_1,\theta)&0\\
                    f_{11}(y_0,\theta)&f_{11}(y_1,\theta)&0\\
                    f_{12}(y_0,\theta)&f_{12}(y_1,\theta)&1
                \end{vmatrix}
                \neq0 ,$$
and it follows that the vector $(0,0,1)^T$ is an element of the set
$A^*$ defined in (\ref{astern}). Therefore we obtain from the
results of Kiefer and Wolfowitz (1965), that the local $D_1$-optimal design is supported on the entire set of Chebyshev-points $\{u_0^*,u_1^*,u_2^*\}$ of the Chebyshev-system $\{f_{10}(u,\theta),f_{11}(u,\theta),f_{12}(u,\theta)\}$. If the support points are given, say
 $u_0,u_1,u_2$ the corresponding
weights can be determined by
    \eqref{gewichtsvektor} such that  the function defined in \eqref{ccrit} is maximal. \\
Consequently the $D_1$-optimality criterion can be expressed as a function of the points
 $u_0,u_1,u_2$, which
 will now be optimized analytically. For this purpose we obtain by a tedious  computation
\begin{equation}\label{D1tomax3}
    \begin{aligned}
        T(\tilde{u},\theta):=\frac{|{M(\xi,\theta)}|}{|{\tilde{M}(\xi,\theta)}|}
                =&u_0^2(u_0-u_1)^2u_1^2(u_1-u_2)^2u_2^2/N
    \end{aligned}
\end{equation}
where ${\tilde{M}(\xi,\theta)}$ denotes the matrix obtained from $M(\xi,\theta)$
by deleting the last row and column,
$\tilde{u}=(u_0,u_1,u_2), \theta=(\theta_0,\theta_1,\theta_2)$ and
\begin{equation*}
    \begin{aligned}
        N=\;&(4\theta_0 u_0 u_1(\theta_1+\theta_2 u_1)u_2+\theta_0^2(u_1(u_2-u_1)+u_0(u_1+u_2))\\
        &+u_0 u_1 u_2(2\theta_1^2 u_1+2\theta_1 \theta_2(u_0(u_1-u_2)+u_1(u_1+u_2))\\
        &+ \theta_2^2(u_0^2(u_1-u_2)+u_0(u_1-u_2)u_2+u_1(u_1^2+u_2^2))))^2.
    \end{aligned}
\end{equation*}
The support points of the $D_1$-optimal design are obtained by maximizing the function
$T(\tilde{u},\theta)$ with respect to $u_0,u_1,u_2$. The necessary conditions
for a maximum yield the
following system of nonlinear equations
            \begin{align}
                \frac{\partial T}{\partial u_0}(\tilde{u},\theta)=\;
                &4\theta_0u_0^2u_1(\theta_1+\theta_2u_1)u_2+\theta_0^2(-2u_0u_1(u_1-u_2)\nonumber\\
                &+u_1^2(u_1-u_2)+u_0^2(u_1+u_2))+u_0^2u_1u_2(2\theta_1^2u_1+4\theta_1\theta_2u_1^2\nonumber\\
                &+\theta_2^2(2u_0u_1(u_1-u_2)+u_0^2(u_2-u_1)+u_1^2(u_1+u_2))) \cdot R_1=0 ,\label{D1abl31}\\
                \frac{\partial T}{\partial u_1}(\tilde{u},\theta)=\;
                &4\theta_0u_0u_1^2u_2(-(\theta_2u_1^2)+\theta_2u_0u_2+\theta_1(u_0-2u_1+u_2))\nonumber\\
                &+\theta_0^2(u_1^2(u_1-u_2)^2-2u_0u_1(u_1^2+u_1u_2-u_2^2)+u_0^2(u_1^2+2u_1u_2-u_2^2))\nonumber\\
                &-u_0u_1^2u_2(2\theta_1^2(u_1^2-u_0u_2)+4\theta_1\theta_2u_1(u_0(u_1-2u_2)+u_1u_2)\nonumber\\
                &+\theta_2^2(u_0^2(u_1-u_2)^2+2u_0u_1(u_1^2-u_1u_2-u_2^2)\nonumber\\
                &+u_1^2(-u_1^2+2u_1u_2+u_2^2))) \cdot R_2=0 ,\label{D1abl32}\\
                \frac{\partial T}{\partial u_2}(\tilde{u},\theta)=\;
                &4\theta_0u_0u_1(\theta_1+\theta_2u_1)u_2^2
                +\theta_0^2(u_1(u_1-u_2)^2+u_0(-u_1^2+2u_1u_2+u_2^2))\nonumber\\
                &+u_0u_1u_2^2(2\theta_1^2u_1+4\theta_1\theta_2u_1^2
                +\theta_2^2(u_0(u_1-u_2)^2\nonumber\\
                &+u_1(u_1^2+2u_1u_2-u_2^2))) \cdot R_3=0 ,\label{D1abl33}
            \end{align}
where  $R_1,R_2$ and $R_3$  are rational functions, which do not
 vanish  for all $u_0,u_1,u_2$ with $0<u_0<u_1<u_2$.
    In order to solve this system of equations, we assume that
\begin{equation}\label{proportional}
 u_0=\frac{u_1}{r}~,~~
 u_2=r\cdot u_1
 \end{equation}
  holds for some factor  $r >1$, which will be specified later.
    Inserting this expression in \eqref{D1abl32} provides as the only positive  solution
        $u_1^*=\sqrt{\theta_0/\theta_2}$.
    Substituting this term into \eqref{D1abl31} or \eqref{D1abl33} yields
the following equation for the factor $r$
        \begin{equation*}
            2\theta_1(\theta_1+4\sqrt{\theta_0\theta_2})r^2-\theta_0\theta_2(1-4r-2r^2-4r^3+r^4)=0
        \end{equation*}
with four roots given by
        \begin{equation}\label{D1nst3}
            r_{1/2}=1 \pm \frac{(2 + \gamma)}{\sqrt{2}} \pm \sqrt{2(1+\sqrt{2}) + (2+\sqrt{2})\gamma + \frac{\gamma^2}{2}},
        \end{equation}
        \begin{equation*}
            r_{3/4}=1 \pm \frac{(2 + \gamma)}{\sqrt{2}} \mp \sqrt{2(1+\sqrt{2}) + (2+\sqrt{2})\gamma + \frac{\gamma^2}{2}},
        \end{equation*}
where  $\gamma=\sqrt{\theta_1\theta_2}^{-1}$.
    The factor $r$ has to be strict greater 1 according to our assumption on the relation between $u_0,u_1$ and $u_2$. This provides only the first solution in \eqref{D1nst3} and
the geometric scaling factor is given by (\ref{scalefac}).
Therefore it remains to justify assumption (\ref{proportional}), which will be done in the second part
of the proof.

\medskip

{\bf (B)} Because the calculation of the support points  $\rho^{-1}\sqrt{\theta_0/\theta_2}$, $\sqrt{\theta_0/\theta_2}$
and $\rho\sqrt{\theta_0/\theta_2}$ in step (A) is based on assumption (\ref{proportional}), we still have prove that these points are
the support points of  the local $D_1$-optimal design. For this purpose  we show  that the unique oscillating polynomial defined by
\eqref{eindpolynom} attends minima and maxima exactly in these support points. Recall that the vector of the partial derivatives of the
regression function $ f_{1}(u,\theta)=\left(f_{10}(u,\theta),f_{11}(u,\theta),f_{12}(u,\theta)\right) $ is given by (\ref{grad3}).
We now define a polynomial $t(u)$ by
\begin{equation}\label{anderespolynom}
t(u)=f_{10}(u,\theta)+\alpha_1f_{11}(u,\theta)+\alpha_2f_{12}(u,\theta)
\end{equation}
and determine the factors $\alpha_1$ and $\alpha_2$  such
that it is equioscillating, i.e.
\begin{subequations}
    \begin{align}
        &t'(u_i^*)=0 &i=0,1,2 \label{erstebed}\\
        &t(u_i^*)=c(-1)^{i-1}&i=0,1,2\label{zweitebed}
    \end{align}
\end{subequations}
for some constant $c\in \R$.
By this choice  the polynomial $t(u)$ must be proportional to the polynomial $\phi(u)$
defined in (\ref{eindpolynom}). For the determination of the coefficients
we  differentiate the polynomial $t(u)$ and get
\begin{equation}\label{ableitt}
t'(u)=\frac{-(\theta_0(1+2u\alpha_1+3u^2\alpha_2))+u(\theta_1(1-u^2\alpha_2)+\theta_2u(3+2u\alpha_1+u^2\alpha_2))}{(\theta_0+u(\theta_1+\theta_2u))^3}
\end{equation}
Substituting the support points $u_1^*=\sqrt{\theta_0/\theta_2}$ and $u_2^*=\rho\sqrt{\theta_0/\theta_2}$ in  \eqref{ableitt} we obtain
from (\ref{erstebed}) two equations
    \begin{eqnarray*}
        0 & =& \frac{\sqrt{\theta_0}(\theta_1+2\sqrt{\theta_0\theta_2})(\theta_2-\theta_0\alpha_2)}
        {\sqrt{\theta_2}\theta_2}\\
        0 & =& \frac{\sqrt{\theta_0}\theta_2(\theta_1\rho+\sqrt{\theta_0\theta_2}
        (3\rho^2-1)+2\theta_0\rho(\rho^2-1)\alpha_1)}{\sqrt{\theta_2}\theta_2}\\ &&+ \frac{\sqrt{\theta_0}\theta_0\rho^2(-(\theta_1\rho)+
  \sqrt{\theta_0\theta_2}(\rho^2-3))\alpha_2}{\sqrt{\theta_2}\theta_2}.
    \end{eqnarray*}
The solution with respect to  $\alpha_1$ and $\alpha_2$ is given by
\begin{equation*}
    \alpha_1=-\frac{\sqrt{\theta_0\theta_2}-\theta_1\rho+\sqrt{\theta_0\theta_2}\rho^2}{2\theta_0\rho}~,
~~\alpha_2=\frac{\theta_2}{\theta_0},
\end{equation*}
which yields for the polynomial $t(u)$ and its derivate
\begin{align}
t(u)\;\;=\;\;&\frac{u(-2\theta_0\rho+\sqrt{\theta_0\theta_2}(1+\rho^2)u-\rho u(\theta_1+2\theta_2u))}{2\theta_0\rho(\theta_0+u(\theta_1+\theta_2u))^2}, \nonumber\\
t'(u)\;\;=\;\;& -\frac{(\sqrt{\theta_0}-\sqrt{\theta_2}u)(\sqrt{\theta_0}\rho-\sqrt{\theta_2}u)(\sqrt{\theta_0}+\sqrt{\theta_2}u)(\sqrt{\theta_0}-\sqrt{\theta_2}\rho u)}{\theta_0\rho(\theta_0+u(\theta_1+\theta_2u))^3}\label{ableittfertig}
\end{align}
respectively. A straightforward calculation shows that
the third support point $u_0^*=\rho^{-1}\sqrt{\theta_0/\theta_2}$
 satisfies $t^\prime (u_0^*)=0$ and  that the three equations in   \eqref{zweitebed}
are  satisfied. Therefore
 it only remains to prove, that the inequality
  $\abs{t(u)}\leq c$ holds on the interval $[0,\infty)$.
In this case  the polynomial $t(u)$ must be proportional to the equioscillating polynomial $\phi(u)$
and the  design with support points
 $\rho^{-1}\sqrt{\theta_0/\theta_2}$, $\sqrt{\theta_0/\theta_2}$ and $\rho\sqrt{\theta_0/\theta_2}$ and optimal
weights is local $D_1$-optimal. \\
Observing the representation  \eqref{ableittfertig} shows that the equation $t'(u)=0$ is equivalent to
\begin{equation}
(\sqrt{\theta_0}-\sqrt{\theta_2}u)(\sqrt{\theta_0}\rho-\sqrt{\theta_2}u)(\sqrt{\theta_0}+\sqrt{\theta_2}u)(\sqrt{\theta_0}-\sqrt{\theta_2}\rho u)=0
\end{equation}
with roots
    $$
    n_0=-\sqrt{\frac{\theta_0}{\theta_2}},\;
    n_1=\frac{1}{\rho}\sqrt{\frac{\theta_0}{\theta_2}},\;
    n_2=\sqrt{\frac{\theta_0}{\theta_2}}\text{ and }
    n_3=\rho\sqrt{\frac{\theta_0}{\theta_2}}.$$
Therefore the function $t(u)$ has exactly three extrema on $\R^+$. Furthermore if
$u\rightarrow\infty$, we have
$t(u)\rightarrow0$  and it follows  that $\abs{t(u)}\leq c$ holds for all $u\geq 0$.
 Consequently, the functions $t(u)$ and $\phi(u)$ are proportional and  the points $u_0^*=\rho^{-1}\sqrt{\theta_0/\theta_2}$, $u_1^*=\sqrt{\theta_0/\theta_2}$, $u_2^*=\rho\sqrt{\theta_0/\theta_2}$ are the support points of
 the local $D_1$-optimal design.
The explicit construction of the weights $w_0$ and $w_1$ is obtained  by substituting the support
points $u_0^*$, $u_1^*$ and $u_2^*$ into \eqref{gewichtsvektor}.\\[2ex]
{\bf (C)} We finally consider  the cases \eqref{D1fall31}, \eqref{D1fall32} and \eqref{D1fall33} in the second part of Theorem \ref{theoD13},
which correspond to a bounded design space. For the sake of brevity we restrict ourselves to the case \eqref{D1fall31}, all other cases are
treated similarly. Obviously the assertion follows from  the existence of a point $u_0^*>0$, such that the function $T(u_0,u_1^*,u_2^*,\theta)$
is increasing in $u_0$ on the interval $(0,u_0^*)$ and decreasing on
    $(u_0^*,u_1^*)$. \\
For a proof of this property  we fix $u_1,u_2$, and note that the function  $\bar{T}(u_0):=T(u_0,u_1,u_2,\theta)$ has minima in $u_0=0$ and
$u_0=u_1$, since the inequality $\bar{T}(u_0)\geq 0$ holds for all
    $u_0\in[0,u_1]$ and $\bar{T}(0)=\bar{T}(u_1)=0$.  Because $\bar{T}(u_0)$ is not constant, there is
at least one maximum in  the interval $(0,u_1)$.  In order to prove that there
 is exactly one maximum, we calculate
        \begin{equation}\label{gleich1}
            \bar{T}'(u_0)=\frac{\partial T}{\partial u_0}(u_0,u_1,u_2,\theta)
            =2u_0(u_0-u_1)u_1^2(u_1-u_2)^2u_2^2\frac{P_4(u_0)}{P_{9}(u_0)},
        \end {equation}
where $P_{9}$ is a polynomial of degree 9 (which is in the following discussion without interest) and the polynomial $P_4$ in the numerator is
given by
    \begin{equation*}
        \begin{aligned}
            P_4(u_0)=\;&4\theta_0u_0^2u_1(\theta_1+\theta_2u_1)u_2
            +\theta_0^2(-2u_0u_1(u_1-u_2)+u_1^2(u_1-u_2)\\&+u_0^2(u_1+u_2))
            +u_0^2u_1u_2(2\theta_1^2u_1+4\theta_1\theta_2u_1^2+\theta_2^2(2u_0u_1(u_1-u_2)\\&+u_0^2
            (u_2-u_1)+u_1^2(u_1+u_2))).
        \end{aligned}
    \end{equation*}
The roots of the function $\bar{T}'$ are given by the roots of the polynomial $P_4$. Differentiating this polynomial yields the function
        \begin{equation*}
            \begin{aligned}
                \frac{\partial P_4}{\partial u_0}(u_0)=\;&
                8\theta_0u_0u_1(\theta_1+\theta_2u_1)u_2+2\theta_0^2(u_1(u_2-u_1)+u_0(u_1+u_2))\\
                &+2u_0u_1u_2(2\theta_1^2u_1+4\theta_1\theta_2u_1^2+
                \theta_2^2(-2u_0^2(u_1-u_2)\\&+3u_0u_1(u_1-u_2)+u_1^2(u_1+u_2))),
            \end{aligned}
        \end{equation*}
which has only one real root. Consequently   $P_4(u_0)$ has just one extremum and therefore at most two roots.
The case of no roots has been excluded above. If  $P_4(u_0)$ would have  two roots, then the function $\bar{T}(u_0)$
has at most two extrema in the interval $(0,u_1)$. However, the function $\bar{T}(u_0)$ is equal to zero in the
 two points $0$ and $u_1$ and in the interval $(0,u_1)$ strictly positive. Therefore the number of its
extrema has to be odd and $\bar{T}(u_0)$ has exactly one maximum on $(0,u_1)$,
which is attained for given
    $(u_1,u_2)=(u_1^*,u_2^*)$ at  a point $u_0^*\in(0,u_1^*)$.\\
Assume that the design space is of the form
$\U=[s,t]$. If the inequality $s<u_0^*$ holds, \eqref{D1optpar3} remains the local $D_1$-optimal  design. However
 if  the inequality $s>u_0^*$ holds, the function $\bar{T}(u_0)$ is maximal in $s$, and it follows that \eqref{D1fall31} is the local $D_1$-optimal  design.\\
\end{proof}
\begin{Remark}
{\rm Note that part (A) of the  proof essentially follows  the arguments presented in Haines (1992) for the $D$-optimality criterion, who
considered the model
\begin{equation*}
\eta(u,\alpha,\beta_0,\beta_1,\beta_2)=\frac{u+\alpha}{\beta_0+\beta_1(u+\alpha)+\beta_2(u+\alpha)^2}.
\end{equation*} However,
the proof presented by Haines (1992)
is not complete, because she  did neither justify the use of the geometric design, nor proves that the system
of necessary conditions has only one solution. In this paper we present a tool for closing  this gap, as demonstrated in part (B) of the
preceding proof. \\
It is also worthwhile to mention that an analogue of Theorem \ref{theoD13} does not hold in
 the four-parameter model
discussed in  Haines (1992). For example  if  $\beta_0=\beta_2=1$, $\beta_1=-1.8$ and $\alpha=0.1$  we obtain by numerical computation that
   the local $D_1$-optimal design is supported at the  Chebyshev-points $\set{0,0.6272,0.9861,1.8714}$ and there  does not exist
   a similar geometric spacing behaviour as in the models considered in this paper. }
\hfill $\Box$
\end{Remark}

The following theorem states the corresponding results for the inverse quadratic regression model with parameterization (\ref{para1}). The
proof is similar to the proof of the previous theorem and therefore omitted.

\medskip
\begin{Theorem}\label{theoD11}
    The local $D_1$-optimal  design $\xi^*_{D_1}$ for the inverse quadratic regression
model (\ref{para1}) on the design space $\U=[0,\infty)$ is given by
        \begin{equation}\label{D1optpar1}
            \xi^*_{D_1}=
                \begin{pmatrix}
                    \frac{1}{\rho}\sqrt{\frac{\theta_1}{\theta_2}}&\sqrt{\frac{\theta_1}{\theta_2}}&\rho \sqrt{\frac{\theta_1}{\theta_2}} \\
                    w_0& w_1& 1-w_0-w_1
                \end{pmatrix}
        \end{equation}
    with
        \begin{align*}                           w_0=&(\sqrt{\theta_2}\theta_1+\sqrt{\theta_1}\rho+\sqrt{\theta_2}\theta_1\rho^2)^2(1+\sqrt{\theta_1\theta_2}(1+\rho))/[(1+\rho)\lambda]\\
          w_1=&(2\theta_1+\sqrt{\theta_1\theta_2})^2\rho(\rho+\sqrt{\theta_1\theta_2}(1+\rho^2))/\lambda
        \end{align*}
        and
        \begin{align*}
        \lambda=\;&\theta_1(\rho(2\rho+3\sqrt{\theta_1\theta_2}(1+\rho)^2)
       +\theta_1\theta_2(1+2\sqrt{\theta_1\theta_2}(1+\rho)^2(1+\rho^2)\\
        &+\rho(8+\rho(6+\rho(8+\rho))))).
        \end{align*}
The geometric scaling factor $\rho$ is given  by (\ref{scalefac})
    with $\gamma=(\sqrt{\theta_1\theta_2})^{-1}$.
      This design is also local $D_1$-optimal on the design space ${\cal U}=[s,t]$ ($0<s<t$), if the inequalities
      $0\leq s \leq
    \rho^{-1}\sqrt{\theta_1/\theta_2}$
    and $t \geq \rho \sqrt{\theta_1/\theta_2}$ are satisfied.\\
 The local $D_1$-optimal  design on the design space ${\cal U}=[s,t]$ for
the inverse quadratic regression model (\ref{para1}) is of the form (\ref{D1fall31})
if the inequalities $s \geq \rho^{-1}\sqrt{\theta_1/\theta_2}$ and      $t>\rho\sqrt{\theta_1/\theta_2}$ hold, of the form (\ref{D1fall32})
if the inequalities $s < \rho^{-1}\sqrt{\theta_1/\theta_2}$ and  $t\leq                     \rho\sqrt{\theta_1/\theta_2}$ are satisfied
 and is of the form (\ref{D1fall33})
if the inequalities $s \geq \rho^{-1}\sqrt{\theta_1/\theta_2}$ and $t \leq \rho\sqrt{\theta_1/\theta_2}$ hold.
\end{Theorem}

In the following discussion we concentrate on the problem of extrapolation
in the inverse quadratic regression model. An optimal design for this purpose
minimizes the variance of the estimate of the expected response
at a point $x_e$ and is therefore $c$-optimal for the vector
$c_e=f_{1}(x_e,\theta)$ in the case of  parameterization (\ref{para3}), and for
the vector
$c_e=f_{2}(x_e,\theta)$ in the case of parameterization (\ref{para1}), respectively.
If $x_e$ is an element of the design space ${\cal U}$ it is obviously
optimal to take all observations at the point $x_e$, and therefore we assume for the remaining
part of this section that  ${\cal U} =[s,t]$, where $0\le s < t$ and $0< x_e <s$ or $x_e>t$.
The following result specifies local optimal extrapolation designs
for the inverse quadratic regression model
which are called local $c_e$-designs in the following discussion.
The proofs are similar to the proofs for $D_1$-optimality  and therefore
omitted.

\medskip
\begin{Theorem}\label{theoc_e3}
Assume that ${\cal U} =[s,t]$, where $0\le s < t$ and $0< x_e<s$ or $x_e>t$, and let $\rho$ denote the
 geometric scaling factor defined in (\ref{scalefac}) with $\gamma=\theta_1(\sqrt{\theta_0\theta_2})^{-1}$.
  If $0\leq s \leq \rho^{-1}
                \sqrt{\theta_0/\theta_2}$ and $ t \geq \rho
                \sqrt{\theta_0/\theta_2}$,
                then the local $c_e$-optimal  design $\xi^*_{c_e}$ for the inverse quadratic regression
model (\ref{para3}) is given by
        \begin{equation}\label{c_eoptpar3}
            \xi^*_{c_e}=
                \begin{pmatrix}
                    \frac{1}{\rho}\sqrt{\frac{\theta_0}{\theta_2}}&\sqrt{\frac{\theta_0}{\theta_2}}&\rho \sqrt{\frac{\theta_0}{\theta_2}} \\
                    w_0& w_1& 1-w_0-w_1
                \end{pmatrix}
        \end{equation}
where
        \begin{align*}
            w_0=&\Big(\sqrt{\theta_0}-x_e\sqrt{\theta_2}\Big)\Big(-x_e\sqrt{\theta_2}+\sqrt{\theta_0}
                \rho\Big)
                \Big(\theta_0\sqrt{\theta_2}+\theta_1\sqrt{\theta_0}\rho + \theta_0\sqrt{\theta_2}\rho^2\Big)^2\\
                &\times((1+\rho)\lambda)^{-1}\\
        w_1=&\Big(2\theta_0\sqrt{\theta_2}+\theta_1\sqrt{\theta_0}\Big)^2\rho
            \Big(-x_e\sqrt{\theta_2}+\sqrt{\theta_0}\rho\Big)
            \Big(\sqrt{\theta_0}-x_e\sqrt{\theta_2}\rho\Big)/\lambda\\
        \end{align*}
        with
        \begin{align*}
        	\lambda=&\;\theta_0\Big(\theta_0^2\theta_2\big(1+6\rho^2+\rho^4\big)+
                \theta_0\Big(2\theta_1^2\rho^2+2\theta_1\rho\big(\sqrt{\theta_0\theta_2}
                \big(1+\rho\big)^2
                -4x_e\theta_2\big(1+\rho^2\big)\big)\\
                &+\theta_2x_e\big(-2\sqrt{\theta_0\theta_2}
                \big(1+\rho\big)^2\big(1+\rho^2\big)+x_e\theta_2\big(1+6\rho^2+\rho^4\big)\big)\Big)\\
                &+ \theta_1x_e\rho\Big(2\sqrt{\theta_0\theta_2}\theta_2x_e
                \big(1+\rho\big)^2-\theta_1\big(\sqrt{\theta_0\theta_2}
                +\rho\big(-2x_e\theta_2+\sqrt{\theta_0\theta_2}\big(2+\rho\big)\big)\big)\Big)\Big)
         \end{align*}
The local $c_e$-optimal  design for the inverse quadratic
    model (\ref{para3}) is of the form (\ref{D1fall31})
if the inequalities $s \geq \rho^{-1}\sqrt{\theta_0/\theta_2}$ and $t > \rho
              \sqrt{\theta_0/\theta_2}$ hold, of the form (\ref{D1fall32})
if the inequalities $s < \rho^{-1}\sqrt{\theta_0/\theta_2}$ and $t \leq \rho
                \sqrt{\theta_0/\theta_2}$ are satisfied  and of the form
    (\ref{D1fall33})
if the inequalities $s \geq \rho^{-1}\sqrt{\theta_0/\theta_2}$ and $t \leq \rho \sqrt{\theta_0/\theta_2}$ hold.
\end{Theorem}

\medskip
\begin{Theorem}
Assume that ${\cal U} =[s,t]$, where $0\le s < t$ and $0< x_e<s$ or $x_e>t$, and let $\rho$ denote the
 geometric scaling factor $\rho$ defined in (\ref{scalefac})    with $\gamma=(\sqrt{\theta_1\theta_2})^{-1}$.
  If $0\leq s \leq \rho^{-1}
                \sqrt{\theta_1/\theta_2}$ and $ t \geq \rho
                \sqrt{\theta_1/\theta_2}$,
                then the  local $c_e$-optimal  design $\xi^*_{c_e}$ for the inverse quadratic regression
model (\ref{para1}) on the design space $\U=[0,\infty)$ is given by
        \begin{equation}
            \xi^*_{c_e}=
                \begin{pmatrix}
                    \frac{1}{\rho}\sqrt{\frac{\theta_1}{\theta_2}}&\sqrt{\frac{\theta_1}{\theta_2}}&\rho \sqrt{\frac{\theta_1}{\theta_2}} \\
                    w_0& w_1& 1-w_0-w_1
                \end{pmatrix}
        \end{equation}
    with
        \begin{align*}
            w_0=&\Big(\sqrt{\theta_1}-x_e\sqrt{\theta_2}\Big)\Big(-x_e\sqrt{\theta_2}+\sqrt{\theta_1}
            \rho\Big)\Big(\theta_1\sqrt{\theta_2}+\sqrt{\theta_1}\rho+\theta_1\sqrt{\theta_2}\rho^2\Big)^2\\
            &\times((1+\rho\big)\lambda)^{-1}\\
        w_1=&\Big(2\theta_1\sqrt{\theta_2}+\sqrt{\theta_1}\Big)^2\rho\Big(-x_e\sqrt{\theta_2}+
            \sqrt{\theta_1}\rho\Big)\Big(\sqrt{\theta_1}-x_e\sqrt{\theta_2}\rho\Big)
            /\lambda\\
        \end{align*}
      with
      	\begin{align*}
      		\lambda =\;&\theta_1\Big(\theta_1^2\theta_2\big(1+6\rho^2+\rho^4\big)+x_e
      		\rho\Big(-\sqrt{\theta_1\theta_2}
      	  +2\sqrt{\theta_1\theta_2}\theta_2x_e\big(1+\rho\big)^2\\
      	  &-\rho\big
  				(-2x_e\sqrt{\theta_2}+\sqrt{\theta_1\theta_2}\big(2+\rho\big)\big)\Big)
  			+\theta_1\Big(2\rho^2+2\rho\big(\sqrt{\theta_1\theta_2}\big(1+\rho\big)^2\\
  				&-4x_e\sqrt{\theta_2}\big(1+\rho^2\big)\big)+x_e
                \big(-2\sqrt{\theta_1\theta_2}\theta_2\big(1+\rho\big)^2\big(1+\rho^2\big)\\
                &+x_e\theta_2^2\big(1+6\rho^2+\rho^4\big)\big)\Big)\Big)
        \end{align*}
    If the design space is given by a finite interval $[s,t]$, $0<s<t$, then the local $c_e$-optimal  design for
    model (\ref{para3}) is of the form (\ref{D1fall31}),
if the inequalities $s \geq \rho^{-1}\sqrt{\theta_1/\theta_2}$ and $t > \rho
              \sqrt{\theta_1/\theta_2}$ hold, of the form (\ref{D1fall32}),
if the inequalities $s < \rho^{-1}\sqrt{\theta_1\theta_2}$ and $t \leq \rho
                \sqrt{\theta_1/\theta_2}$ are satisfied,  and of the form (\ref{D1fall33})
if the inequalities $s \geq \rho^{-1}\sqrt{\theta_1/\theta_2}$ and $t \leq \rho \sqrt{\theta_1/\theta_2}$ hold.
\end{Theorem}

 Note that for  a sufficiently large design interval
 all designs presented in this section are supported at the same points, the
 Chebyshev points corresponding to the Chebyshev system of  the components
 of the gradient of the regression function. In the next section we will demonstrate that these
 points are also the support points of the local $E$-optimal design for the inverse quadratic regression
 model.\par

\bigskip
 \setcounter{chapter}{4}
\setcounter{equation}{0} 
\setcounter{Theorem}{0}
\noindent {\bf 4. Local $D$- and $E$-optimal designs}
We begin stating the corresponding result for the $D$-optimality criterion. The proof is omitted
because it requires
arguments which are similar  as those presented in Haines (1992) and in the proof
of Theorem \ref{theoD13}.

\medskip
\begin{Theorem}\label{dpara3}
The local $D$-optimal  design $\xi^*_{D}$ for the inverse quadratic regression
model (\ref{para3}) on the design space $\U=[0,\infty)$ is given by
        \begin{equation}\label{Doptpar3}
            \xi^*_{D}=
                \begin{pmatrix}
                    \frac{1}{\rho}\sqrt{\frac{\theta_0}{\theta_2}}&\sqrt{\frac{\theta_0}{\theta_2}}&\rho \sqrt{\frac{\theta_0}{\theta_2}} \\
                    \frac{1}{3}&\frac{1}{3} &\frac{1}{3}
                \end{pmatrix}
        \end{equation}
with the geometric scaling factor
        \begin{equation}\label{geo}
                \rho=\frac{\delta+\sqrt{\delta^2-4}}{2},
        \end{equation}
where the constants $\delta$ and $\gamma$ are defined by $\delta=(1/2)(\gamma+1+\sqrt{\gamma^2+6\gamma+33})$ and
$\gamma=\theta_1(\sqrt{\theta_0\theta_2})^{-1}$, respectively.   This design is also local $D$-optimal on the design space ${\cal U}=[s,t]$ ($0<s<t$), if the inequalities $0\leq s \leq
    \rho^{-1}\sqrt{\theta_0/\theta_2}$
    and $t \geq \rho \sqrt{\theta_0/\theta_2}$ are satisfied.\\
 The local $D$-optimal  design on the design space ${\cal U}=[s,t]$ for the
inverse quadratic regression model (\ref{para1}) is of the form (\ref{D1fall31}),
if the inequalities $s \geq \rho^{-1}\sqrt{\theta_1/\theta_2}$ and $t > \rho
                \sqrt{\theta_1/\theta_2}$ hold, of the form
(\ref{D1fall32}),
if the inequalities $s < \rho^{-1}\sqrt{\theta_1/\theta_2}$ and $t \leq \rho
                \sqrt{\theta_1/\theta_2}$ are satisfied,  and is of the form (\ref{D1fall33}),
if the inequalities $s \geq \rho^{-1}\sqrt{\theta_1/\theta_2}$ and $t \leq \rho
                \sqrt{\theta_1/\theta_2}$ hold.
\end{Theorem}

\medskip
\begin{Theorem}
The local $D$-optimal  design $\xi^*_{D}$ for the inverse quadratic regression
model (\ref{para1}) on the design space $\U=[0,\infty)$ is given by
        \begin{equation}\label{Doptpar1}
            \xi^*_{D}=
                \begin{pmatrix}
                    \frac{1}{\rho}\sqrt{\frac{\theta_1}{\theta_2}}&\sqrt{\frac{\theta_1}{\theta_2}}&\rho \sqrt{\frac{\theta_1}{\theta_2}} \\
                    \tfrac{1}{3}&\tfrac{1}{3} &\tfrac{1}{3}
                \end{pmatrix}
        \end{equation}
with the geometric scaling factor $\rho$ is given by (\ref{geo})
with $\gamma=(\sqrt{\theta_1\theta_2})^{-1}$.
  This design is also $D$-optimal on the design space ${\cal U}=[s,t]$ ($0<s<t$), if
  the inequalities $0\leq s \leq
    \rho^{-1}\sqrt{\theta_1/\theta_2}$
    and $t \geq \rho \sqrt{\theta_1/\theta_2} $ are satisfied. \\
 The local $D$-optimal  design on the design space ${\cal U}=[s,t]$ for
  the inverse quadratic regression
model (\ref{para1}) is of the form (\ref{D1fall31}),
if the inequalities $s \geq \rho^{-1}\sqrt{\theta_1/\theta_2}$ and $t > \rho
                \sqrt{\theta_1/\theta_2}$ hold, of the form
    (\ref{D1fall32}),
if the inequalities $s < \rho^{-1}\sqrt{\theta_1/\theta_2}$ and $t \leq \rho
                \sqrt{\theta_1/\theta_2}$  are satisfied, and is of the form
    (\ref{D1fall33}),
if the inequalities $s \geq \rho^{-1}\sqrt{\theta_1/\theta_2}$ and $t \leq \rho
                \sqrt{\theta_1/\theta_2}$ hold.
\end{Theorem}

\medskip
We will conclude this section with the discussion of the {$E$-optimality} criterion. For this
purpose recall that a
 design $\xi_E$ is local $E$-optimal if and only if there exists a matrix $E\in \conv(S)$ such that
 the inequality
\begin{equation}
f'(u,\theta)Ef(u,\theta)\leq \lambda_{\min} \qquad \text{  }
\end{equation}
holds for all $u\in\U$,
where $\lambda_{\min}$ denotes the minimum eigenvalue of the matrix $M(\xi_E,\theta)$
and
\begin{equation}
S=\set{zz'\;|\;\|z\|_2=1, ~z \text{ is an eigenvector of } M(\xi_E,\theta) \text{ corresponding to } \lambda_{\min}}.
\end{equation}
The following two results specify the local $E$-optimal designs for the inverse quadratic regression
models with parameterization (\ref{para3}) and (\ref{para1}). Because both statements are proved similarly,
we restrict ourselves to a proof of the first theorem.

\medskip
\begin{Theorem}\label{theoe3}
    The local $E$-optimal  design $\xi^*_{E}$ for the inverse quadratic regression
model (\ref{para3}) on the design space $\U=[0,\infty)$ is given by
        \begin{equation}\label{eoptpar3}
            \xi^*_{E}=
                \begin{pmatrix}
                    \frac{1}{\rho}\sqrt{\frac{\theta_0}{\theta_2}}&\sqrt{\frac{\theta_0}{\theta_2}}&\rho \sqrt{\frac{\theta_0}{\theta_2}} \\
                    w_0& w_1& 1-w_0-w_1
                \end{pmatrix}
        \end{equation}
where the weights $w_0$, $w_1$ are given by (\ref{gewichtsvektor}) and $c$ is the vector with components given by the coefficients of the
Chebyshev polynomial, that is
 \begin{align*}
    c=&\;\Big(-\frac{\sqrt{\theta_0}(2\theta_1^2\rho^2 + 2\sqrt{\theta_0}\theta_1\sqrt{\theta_2}\rho(1 + \rho)^2 + \theta_0\theta_2
            (1 + 6\rho^2 + \rho^4))}{\sqrt{\theta_2}(-1+ \rho)^2\rho},\\[1ex]
        &\;\frac{\theta_1^2\rho(1 + \rho)^2 + 8\sqrt{\theta_0}\theta_1\sqrt{\theta_2}\rho(1 + \rho^2) + 2\theta_0\theta_2(1 + \rho)^2(1 + \rho^2)}
            {(-1 + \rho)^2\rho},\\[1ex]
        &\;-\frac{\sqrt{\theta_2}(2\theta_1^2\rho^2 + 2\sqrt{\theta_0}\theta_1\sqrt{\theta_2}\rho(1 + \rho)^2 + \theta_0\theta_2(1 + 6\rho^2 + \rho^4))}
            {\sqrt{\theta_0}(-1 + \rho)^2\rho}\Big)^T.
 \end{align*}
The geometric scaling factor is given  by (\ref{scalefac})
    with $\gamma=\theta_1(\sqrt{\theta_0\theta_2})^{-1}$.  This design is also local $E$-optimal on the design space ${\cal U}=[s,t]$ ($0<s<t$),
    if the inequalities $0\leq s \leq
    \rho^{-1}\sqrt{\theta_0/\theta_2}$
    and $t \geq \rho \sqrt{\theta_0/\theta_2}$ are satisfied.\\
 The local $E$-optimal  design on the design space ${\cal U}=[s,t]$ for
    model (\ref{para3}) is of the form (\ref{D1fall31}),
if the inequalities $s \geq \rho^{-1}\sqrt{\theta_0/\theta_2}$ and $t > \rho
              \sqrt{\theta_0/\theta_2}$ hold, of the form
    (\ref{D1fall32}),
if the inequalities $s < \rho^{-1}\sqrt{\theta_0/\theta_2}$ and $t \leq \rho
                \sqrt{\theta_0/\theta_2}$  are satisfied, and of the form
    (\ref{D1fall33}),
if the inequalities $s \geq \rho^{-1}\sqrt{\theta_0/\theta_2}$ and $t \leq \rho \sqrt{\theta_0/\theta_2}$ hold.
\end{Theorem}

\begin{proof}
It is straightforward  to show that every subset of $\{f_{10}(u,\theta),f_{11}(u,\theta),$\\$f_{12}(u,\theta)\}$, the components of the vector $f_{1}(u,\theta)$,
 which consists of $2$ elements, is a (weak) Chebychev-system. Therefore it
 follows from Theorem 2.1 in Imhof, Studden (2001) that the local  $E$-optimal is supported at the Chebyshev
 points. The assertion regarding the weights finally follows from  (\ref{gewichtsvektor}) observing
 that the results of Imhof and Studden (2001) imply that the
 local $E$-optimal design is also $c$-optimal for the vector $c$ with components given by the
coefficients of the Chebyshev polynomial.
 \hfill $\Box$

\end{proof}

\begin{Theorem}
    The local $E$-optimal  design $\xi^*_{E}$ for the inverse quadratic regression
model (\ref{para1}) on the design space $\U=[0,\infty)$ is given by
        \begin{equation}
            \xi^*_{E}=
                \begin{pmatrix}
                    \frac{1}{\rho}\sqrt{\frac{\theta_1}{\theta_2}}&\sqrt{\frac{\theta_1}{\theta_2}}&\rho \sqrt{\frac{\theta_1}{\theta_2}} \\
                    w_0& w_1& 1-w_0-w_1
                \end{pmatrix},
        \end{equation}
where the weights $w_0$, $w_1$ are given by (\ref{gewichtsvektor}) and $c$ is the vector with components given by the coefficients of the
Chebyshev polynomial, that is
\begin{align*}
     c=&\;\Big(-1-2\sqrt{\theta_1\theta_2}-\frac{2(2\rho+\sqrt{\theta_1\theta_2}(1+\rho)^2)(\rho+\sqrt{\theta_1\theta_2}(1+\rho^2))}{(-1+\rho)^2\rho},\\[1ex]
         &\;-\frac{\sqrt{\theta_1}(1+2\sqrt{\theta_1\theta_2})(2\rho+\sqrt{\theta_1\theta_2}(1+\rho)^2)(\rho+\sqrt{\theta_1\theta_2}(1+\rho^2))}
            {\theta_0\sqrt{\theta_2}(-1+\rho)^2\rho},\\[1ex]
         &\;-\frac{\sqrt{\theta_2}(1+2\sqrt{\theta_1\theta_2})(2\rho+\sqrt{\theta_1\theta_2}(1+\rho)^2)(\rho+\sqrt{\theta_1\theta_2}(1+\rho^2))}
            {\theta_0\sqrt{\theta_1}(-1+\rho)^2\rho}\Big)^T.
 \end{align*}
The geometric scaling factor is given  by (\ref{scalefac})
    with $\gamma=(\sqrt{\theta_1\theta_2})^{-1}$.   This design is also local $E$-optimal on the design space ${\cal U}=[s,t]$ ($0<s<t$), if
    the inequalities $0\leq s \leq
    \rho^{-1}\sqrt{\theta_1/\theta_2}$
    and $t \geq \rho \sqrt{\theta_1/\theta_2}$ are satisfied.\\
 The local $E$-optimal  design on the design space ${\cal U}=[s,t]$ for
    model (\ref{para3}) is of the form (\ref{D1fall31}),
if the inequalities $s \geq \rho^{-1}\sqrt{\theta_1/\theta_2}$ and $t > \rho
              \sqrt{\theta_1/\theta_2}$ hold, of the form (\ref{D1fall32}),
if the inequalities $s < \rho^{-1}\sqrt{\theta_1/\theta_2}$ and $t \leq \rho
                \sqrt{\theta_1/\theta_2}$  are satisfied, and of the form (\ref{D1fall33}),
if the inequalities $s \geq \rho^{-1}\sqrt{\theta_1/\theta_2}$ and $t \leq \rho \sqrt{\theta_1/\theta_2}$ hold.
\end{Theorem}\par

 \setcounter{chapter}{5}
\setcounter{equation}{0} 
\noindent {\bf 5. Further discussion}

In this Section we discuss some practical aspects of the local optimal designs derived in the previous sections. In particular, we calculate
the efficiency of a design, which has
 recently been used in practice and investigate
the efficiency of local optimal designs with respect to other optimality criteria. Throughout this paper the efficiency of a design $\xi$ is defined by $\eff_\Phi(\xi)=\Phi(\xi)/\sup_{\eta}\Phi(\eta)$, where $\Phi$ denotes the particular optimality criterion under consideration and the optimal design maximizes $\Phi$.

 Landete-Castillejos and Gallego (2000) used the inverse quadratic regression model
 to analyze data, which  were obtained from  lactating red deer hinds
(Cervus elaphus). They concluded that inverse quadratic polynomials with
parameterization (\ref{para3}) can adequately
 describe the common lactation curves. The design space was given by the interval $\mathcal{U} = [1,14]$, and the design used by these authors
 was a uniform design with support points $(1,2,3,4,5,6,10,14)$, which is denoted by
 $\xi_u$ throughout this section. The  estimates for the parameters of model
 (\ref{para3}) are given by  $ \hat{\theta}_0=0.0002865$, $ \hat{\theta}_1=0.0002117$ and $\hat{\theta}_2=0.0000301$.
Table {5.1}  shows the local optimal designs for the different optimality criteria considered in Section 3 and 4,
 where we used the point $x_e=21$ for the calculation of the optimal extrapolation
 design.

\bigskip

\begin{tabularx}{340pt}{|c|>{\centering\arraybackslash}X@{:}>{\centering\arraybackslash}X >{\centering\arraybackslash}X >{\centering\arraybackslash}X|>{\centering\arraybackslash}X|}
\hline
{\bf Criterion}&\multicolumn{4}{c|}{\bf Optimal design} \\
\hline
&points&1&3.4089&14\\
\raisebox{1.5ex}[-1.5ex]{${\bf D}$}&weights&1/3&1/3&1/3\\
\hline
&points&1&3.3561&14\\
\raisebox{1.5ex}[-1.5ex]{${\bf E}$}&weights&0.3972&0.3914&0.2114\\
\hline
&points&1&3.3561&14\\
\raisebox{1.5ex}[-1.5ex]{${\bf D_1}$}&weights&0.1239&0.2884&0.5877\\
\hline
&points&1&3.3561&14\\
\raisebox{1.5ex}[-1.5ex]{${\bf c_e}$}&weights&0.0582&0.1535&0.7883\\
\hline
\end{tabularx}

\medskip\medskip\medskip
\noindent
{\bf Table 5.1.} {\it $D$-, $E$-, $D_1$- and $c_e$-optimal designs for parametrization (\ref{para3}).}

\bigskip

The efficiencies of the different designs are shown in Table {5.2}. We observe that the design of  Landete-Castillejos and Gallego (2000)
 yields rather low efficiencies with respect to all optimality
 criteria, and the efficiency of the statistical analysis could have been improved by allocating observations according to local optimal design [see the first row in Table {5.2}].
 For example a confidence interval based on the local
 $D_1$-optimal design would yield $66\%$ shorter confidence intervals for the
 parameter $c$ as the design actually used by   Landete-Castillejos and Gallego (2000).
 The advantages of the local optimal designs are also clearly visible for the other criteria.

\bigskip

\begin{tabularx}{340pt}{|c|>{\centering\arraybackslash}X|>{\centering\arraybackslash}X|>{\centering\arraybackslash}X|>{\centering\arraybackslash}X|}
\hline
&${\bf D}$&${\bf E}$&${\bf D_1}$&${\bf c_e}$\\
\hline
${\bf \xi_u}$&69.92&50.33&45.85&33.82\\
\hline
${\bf \xi^*_D}$&100 &94.18&75.28&43.60\\
\hline
${\bf \xi^*_E}$&93.96&100&51.89&25.71\\
\hline
${\bf \xi^*_{D_1}}$&74.63&53.05&100&80.40\\
\hline
${\bf \xi^*_{c_e}}$&51.23&33.24&85.73&100\\
\hline
\end{tabularx}

\medskip\medskip\medskip
\noindent {\bf Table 5.2.} {\it Efficiencies of local optimal designs and the uniform design $\xi_u$ for the inverse quadratic model (parameterization (\ref{para3})) with
respect to various alternative criteria (in percent). The design space is the interval $\mathcal{U} = [1,14]$, and the estimates of the
parameters are given by $\hat{\theta}_0=0.0002865$, $\hat{\theta}_1=0.0002117$ and $\hat{\theta}_2=0.0000301$. The local extrapolation optimal design is calculated for
the point $x_e=21$. }

Note that the data is usually used for several purposes, for example for discrimination between a linear and a quadratic inverse polynomial and
for extrapolation using the identified model. Therefore it is important that an optimal design for a specific optimality criterion yields also
reasonable efficiencies with respect to alternative criteria, which reflect other aspects of the statistical analy-\\sis. In Table 5.2 we also compare
the efficiency of a given local optimal design with respect to the other optimality criteria. For example, the local $D$-optimal design has efficiencies $94.18\%$, $75.28\%$ and $43.60\%$ with
respect to the $E$-, $D_1$- and $c_e$-optimality criterion, respectively. Thus this design is rather efficient for the $D_1$- and $E$-optimality
criterion, but less efficient for extrapolation.  The situation for the $D_1$-optimal design is similar, where the role of the $c_e$- and
$E$-criterion have to be interchanged. On the other hand the performance of the local $E$- and $c_e$-optimal design depends strongly on the
underlying optimality criterion. The local $E$-optimal design yields only a satisfactory $D$-efficiency, but is less efficient with respect to
the $c_e$- and $D_1$-optimality criterion, while the local $c_e$-optimal design yields only a satisfactory $D_1$-efficiency.

\newpage
\noindent {\large\bf Acknowledgment}
The support of the Deutsche Forschungsgemeinschaft (SFB 475, \textquotedblleft
Komplexit\"{a}tsreduktion in multivariaten Datenstrukturen\textquotedblright)
is gratefully acknowledged. The work of the authors was also supported in part by
an NIH grant award IR01GM072876:01A1 and by the BMBF project SKAVOE.
 The authors are also grateful to two referees and the associate editor for their constructive
 comments on an earlier version of this paper and to M.\ Stein,
who typed parts of this paper with considerable technical expertise.
\par

\bigskip
\noindent{\large\bf References}
\begin{description}
\item
 Cobby, J. M., Chapman, P. F. and Pike, D.J. (1986). Design of experiments for estimating inverse quadratic polynomial responses. {\it Biometrics} {\bf 42}, 659-664.

\item
    Dette, H. (1995). Optimal designs for identifying the degree of a polynomial regression. {\it Annals of Statistics} {\bf 23}, 1248-1267.

\item
    Dette, H. and Biedermann, S. (2003). Robust and efficient designs for the Michaelis-Menten model. {\it Journal of the American Statistical Association} {\bf 98}, 679-686.

\item
    Dette, H., Melas, V. B. and Pepelyshev, A. (2003). Standardized maximin-optimal designs for the Michaelis-Menten model. {\it Statistica Sinica}, {\bf 13} 1147-1163.


\item
Haines, L. M. (1992). Optimal design for inverse quadratic polynomials. {\it South African Statistical Journal} {\bf 26}, 25-41.

\item
    Imhof, L. A. and Studden, W. J. (2001). $E$-optimal designs for rational models. {\it Annals of Statistics} {\bf 29}, No. 3, 763-783.

\item
    Jennrich, R. I. (1969). Asymptotic properties of nonlinear regression. {\it Annals of Mathmatical Statistics} {\bf 40}, No. 2, 633-643.

\item
    Karlin, S. and Studden, W. J. (1966). {\it Tschebycheff-Systems: with applications in analysis and statistics}. Wiley, New York.

\item
Kiefer, J. (1974). General equivalence theory for optimum designs (approximate theory). {\it Annals of Statistics} {\bf  2}, 849-879.

\item
Kiefer, J. and Wolfowitz, J. (1965). On a theorem of Hoel and Levine on extrapolation designs. {\it Annals of Mathmatical Statistics} {\bf 36}, 1627-1655.

\item
    Landete-Castillejos, T. and Gallego, L. (2000). Technical note: The ability of mathmatical models to describe the shape of lactation curves. {\it Journal of Animal Science} {\bf 78}, 3010-3013.

\item
Lopez-Fidalgo, J. and Wong, W. K. (2002). Design issues for the Michaelis-Menten model. {\it Journal of Theoretical Biology} {\bf 215}, 1-11.

\item
    Nelder, J. A. (1966). Inverse polynomials, a useful group of  multifactor response functions. {\it Biometrics} {\bf 22}, 128-141.

\item
 Pukelsheim, F. (1993). {\it Optimal design of experiments}. John Wiley \& Sons Inc., New York.

\item
Pukelsheim, F. and Torsney, B. (1991). Optimal weights for experimental designs on lineary independent support points. {\it Annals of Statistics} {\bf 19},1614-1625.

\item
    Pukelsheim, F. and Rieder, S. (1992). Efficient rounding of approximate designs. {\it Biometrica} {\bf 79}, 763-770.


\item
    Ratkowsky, D. A. (1990). {\it Handbook of nonlinear regression models}. New York: Marcel Dekker.

\item
    Serchand, L., McNew, R. W., Kellogg, D. W. and Johnson, Z.B. (1995). Selection of a mathmatical model to generate lactation curves using daily milk yields of Holstein cows. {\it Journal of Dairy Science} {\bf 78}, 2507-2513.
\item
Silvey, S. D. (1980). {\it Optimal design}. Chapman and Hall, New York.

\item
    Sparrow, P. E. (1979a). Nitrogen response curves of spring barley. {\it Journal of Agricultural Science} {\bf 92}, 307-317.

\item
Sparrow, P. E. (1979b). The comparison of five response curves for representing the relationship between the annual dry-matter yield of grass herbage and fertilizer nitrogen. {\it Journal of Agricultural Science} {\bf 93}, 513-520.

\item
Song, D. and Wong, W. K. (1998). Optimal designs for the Michaelis-Menten model with heterodastic errors. {\it Communications in Statistics - Theory and Methods} {\bf 27}, 1503-1516.

\item
    Stigler, S. (1971). Optimal experimental design for polynomial regression. {\it Journal of the American Statistical Association} {\bf 66}, 311-318.

\item
    Studden, W. J. (1982). Some robust type $D$-optimal designs in polynomial regression. {\it Journal of the American Statistical Association} {\bf 77}, 916-921.
\end{description}


\vskip .65cm
\noindent
Ruhr-Universität Bochum, Fakultät für Mathematik, 44780 Bochum, Germany
\vskip 2pt
\noindent
E-mail: (holger.dette@ruhr-uni-bochum.de)
\vskip 2pt

\noindent
Ruhr-Universität Bochum, Fakultät für Mathematik, 44780 Bochum, Germany
\vskip 2pt
\noindent
E-mail: (tina.kiss12@googlemail.com)
\vskip .3cm


\end{document}